\documentclass[submission,copyright,creativecommons]{eptcs}
\providecommand{\event}{NCL 2022} 
\usepackage{underscore}           

\usepackage[T1]{fontenc}

\usepackage{amsmath}
\usepackage{amsfonts}
\usepackage{amssymb}
\usepackage{amsthm}
\usepackage{hyperref}
\usepackage{xcolor}

\theoremstyle{definition}
\newtheorem{definition}{Definition}[section]
\newtheorem{theorem}[definition]{Theorem}

\newtheorem{proposition}[definition]{Proposition}
\newtheorem{corollary}[definition]{Corollary}
\newtheorem{example}[definition]{Example}
\newtheorem{remark}[definition]{Remark}

\newcommand{\fm}[1]{\textbf{L}_{#1}}
\newcommand{\sder}{\vdash}

\newcommand{\conn}{\copyright}

\newcommand{\pnmat}[1]{\mathbb{#1}}

\newcommand{\Mt}{\mathbb{M}}
\newcommand{\tuple}[1]{\langle #1 \rangle}

\newcommand{\qi}{q_{\mathsf{init}}}
\newcommand{\inc}[1]{#1^{++}}
\newcommand{\test}[1]{#1^{\mathsf{test}}}

\newcommand{\step}{\mathsf{step}}
\newcommand{\cep}{\epsilon}
\newcommand{\cz}{\mathsf{zero}}
\newcommand{\suc}{\mathsf{suc}}
\newcommand{\enc}{\mathsf{enc}}
\newcommand{\error}{\mathtt{error}}
\newcommand{\conf}{\mathtt{conf}}
\newcommand{\init}{\mathtt{init}}
\newcommand{\n}{\mathtt{r}}
\newcommand{\Num}{\textbf{Rm}}
\newcommand{\seq}{\mathsf{seq}} 
\newcommand{\mC}{\mathcal{C}}
\newcommand{\nxt}{\textsf{nxt}} 
\newcommand{\Conf}{\mathbf{Conf}} 
\newcommand{\LMt}{\Mt_{\text{\L}}}
\newcommand{\Rm}{\textbf{Rm}}
\newcommand{\zero}{\mathsf{zero}}

\DeclareMathOperator{\Var}{\mathsf{Var}}
\DeclareMathOperator{\Sub}{\mathsf{Sub}}
\DeclareMathOperator{\Val}{\mathsf{Val}}

\title{Monadicity of Non-deterministic Logical\\Matrices is Undecidable}
\author{Pedro Filipe \quad Carlos Caleiro \quad S\'ergio Marcelino
\footnote{Research funded by FCT/MCTES through national funds and when applicable co-funded by EU under the project UIDB/50008/2020. The first author acknowledges the grant 
PD/BD/135513/2018 by FCT, under the LisMath PhD programme.}
\institute{SQIG - Instituto de Telecomunica\c{c}\~oes}
\institute{Dep. Matem\'atica, Instituto Superior T\'ecnico,\\ Universidade de Lisboa, Portugal}
\email{pedro.g.filipe@tecnico.ulisboa.pt \quad \{ccal,smarcel\}@math.tecnico.ulisboa.pt}
}
\def\titlerunning{Monadicity of NMatrices is Undecidable}
\def\authorrunning{P. Filipe, C. Caleiro \& S. Marcelino}
\begin{document}
\maketitle

\begin{abstract}
The notion of non-deterministic logical matrix (where connectives are interpreted as multi-functions) preserves many good properties of traditional semantics based on
logical matrices (where connectives are interpreted as functions) whilst finitely characterizing a much wider class of logics, and has proven to be decisive in a myriad
of recent compositional results in logic. Crucially, when a finite non-deterministic matrix satisfies monadicity (distinct truth-values can be separated by unary formulas) one can automatically produce an
axiomatization of the induced logic. Furthermore, the resulting calculi are analytical and enable algorithmic proof-search and symbolic counter-model generation.

For finite (deterministic) matrices it is well known that checking monadicity is decidable. We show that, in the presence of non-determinism, the property becomes 
undecidable. As a consequence, we conclude that there is no algorithm for computing the set of all multi-functions expressible in a given finite Nmatrix.
The undecidability result is obtained by reduction from the halting problem for deterministic counter machines.
\end{abstract}

\section{Introduction}\label{sec:intro}
Logical matrices are arguably the most widespread semantic structures for propositional logics~\cite{WojBook,AlgLogBook}. After {\L}ukasiewicz, 
a logical matrix consists in an underlying algebra, functionally interpreting logical connectives over a set of truth-values, together with a designated set of truth-values. 
The logical models (valuations) are obtained by considering homomorphisms from the free-algebra in the matrix similarity type into the algebra, and formulas that hold in 
the model are the ones that take designated values.

However, in recent years, it has become clear that there are advantages in departing from semantics based on logical matrices, by adopting a non-deterministic generalization 
of the standard notion. Non-deterministic logical matrices (Nmatrices) were introduced in the beginning of this century by Avron and his collaborators~\cite{Avron0,Avron1}, 
and interpret connectives by multi-functions instead of functions. The central idea is that a connective can non-deterministically pick from a set of possible values instead 
of its value being completely determined by the input values. Logical semantics based on Nmatrices are very malleable, allowing not only for finite characterizations of 
logics that do not admit finite semantics based on logical matrices, but also for general recipes for various practical problems in logic~\cite{wollic17}. Further, Nmatrices 
still permit, whenever the underlying logical language is sufficiently expressive, to extend from logical matrices general techniques for effectively producing analytic 
calculi for the induced logics, over which a series of reasoning activities in a purely symbolic fashion can be automated, including proof-search and counter-model 
generation~\cite{SS,Avron0,Avron1,AKZ,CCALJM,synthese,wollic19}. In its simplest form, the sufficient expressiveness requirement mentioned above corresponds to
\emph{monadicity}~\cite{SS,synthese,wollic19}. A Nmatrix is \emph{monadic} when pairs of distinct truth-values can be separated by unary formulas of the logic. This crucial 
property is decidable, in a straightforward way, for finite logical matrices, as one can simply compute the set of all unary functions expressible in a given finite matrix. 
However, the computational character of monadicity with respect to Nmatrices has not been studied before.

In this paper we show that, in fact, monadicity is undecidable for Nmatrices. Our proof is obtained by means of a suitable reduction from the halting problem for counter 
machines, well known to be undecidable~\cite{minsky}. Several details of the construction are inspired by results about the inclusion of infectious values in
Nmatrices~\cite{ismvl}, and also by undecidability results concerning term-dag-automata (a computational model that bears some interesting connections with
Nmatrices)~\cite{ClosureProperties}. As a consequence, we conclude that the set of all unary multi-functions expressible in a given finite Nmatrix is not computable.

The paper is organized as follows. In Section~\ref{sec:prelim} we introduce and illustrate our objects of study, namely logical matrices and Nmatrices, the logics they 
induce, and the monadicity property. Section~\ref{sec:machines} recalls the counter machine model of computation, and shows how its computations can be encoded into suitable 
Nmatrices. Finally, Section~\ref{sec:undecided} establishes our main results, namely the undecidability of monadicity for Nmatrices, and as a corollary the uncomputability of 
expressible multi-functions. We conclude, in Section~\ref{sec:conclude}, with a discussion of the importance of the results obtained and several topics for further research.

\section{Preliminaries}\label{sec:prelim}
In this section we recall the notion of logical matrix, non-deterministic matrix, and their associated logics. We also introduce, exemplify and discuss the notion of 
monadicity.

\paragraph*{Matrices, Nmatrices and their logics}
A signature $\Sigma$ is a family of connectives indexed by their arity, $\Sigma=\{\Sigma^{(k)}:k\in \mathbb{N}\}$.
The set of formulas over $\Sigma$ based on a set of propositional variables $P$ is denoted by $L_\Sigma(P)$.
The set of subformulas (resp, variables) of a formula $\varphi\in L_\Sigma(P)$ is denoted by $\Sub(\varphi)$ (resp.,$\Var(\varphi)$).
There are two subsets of $L_\Sigma(P)$ that will be of particular interest to us: the set of closed formulas, denoted by $L_\Sigma(\emptyset)$, and the set of unary
(or monadic) formulas, denoted by $L_\Sigma(\{p\})$.

A $\Sigma$-Nmatrix, is a tuple $\Mt=\tuple{A,\cdot_\Mt,D}$ where $A$ is the set of truth-values, $D\subseteq A$ is the set of designated truth-values, and for each
$\conn\in \Sigma^{(k)}$, the function $\conn_\Mt:A^k\to \wp(A)\setminus\{\emptyset\}$ interprets the connective $\conn$. A $\Sigma$-Nmatrix $\Mt$ is finite if it contains 
only a finite number of truth-values and $\Sigma$ is finite. Clearly, this definition generalizes the usual definition of logical matrix, which is recovered when, for every 
$\conn\in\Sigma^{(k)}$ and $a_1,\ldots,a_k \in A$, $\conn_{\Mt}(a_1,\ldots,a_k)$ is a singleton. In this case we will sometimes refer to $\conn_{\Mt}$ simply as a function.
A valuation  over $\pnmat{M}$ is a function $v: L_{\Sigma}(P) \rightarrow A$, such that, $v(\conn(\varphi_1, \dots, \varphi_n)) \in \conn_{\pnmat{M}}(v(\varphi_1), \dots,
v(\varphi_n))$ for every $\conn\in \Sigma^{(k)}$.
We use $\Val(\Mt)$ to denote the set of all valuations over $\Mt$.
A valuation $v \in \Val(\Mt)$ is said to satisfy a formula $\varphi$ if $v(\varphi) \in D$, and is said to falsify $\varphi$, otherwise.
Note that every formula $\varphi\in L_{\Sigma}(P)$, with $\Var(\varphi)=\{p_1,\ldots,p_k\}$, defines a multi-function $\varphi_\Mt:A^{k}\to \wp(A)\setminus\{\emptyset\}$ as
$\varphi_\Mt(x_1,\ldots,x_k)=\{v(\varphi):v\in \Val(\Mt),v(p_i)=x_i,1\leq i\leq k\}$. The multi-function $\varphi_\Mt$ is said to be \emph{represented}, or
\emph{expressed}, by the formula $\varphi$ in $\Mt$. Furthermore, we say that a multi-function $f$ is \emph{expressible} in an Nmatrix $\Mt$ if there is a formula $\varphi$ 
such that $\varphi_{\Mt} = f$.

The logic induced by an Nmatrix $\Mt$ is the Tarskian consequence relation $\vdash_\Mt\,\subseteq \wp(L_\Sigma(P))\times L_\Sigma(P)$ defined as $\Gamma\vdash_\Mt \varphi$ 
whenever, for every $v \in \Val(\Mt)$, if $v(\Gamma)\subseteq D$ then $v(\varphi)\in D$. This definition generalizes the usual logical matrix
semantics~\cite{WojBook,AlgLogBook}.
As usual, a formula $\varphi$ is said to be a theorem of $\Mt$ if $\emptyset\vdash_\Mt \varphi$.

\paragraph*{Monadicity}
Given a $\Sigma$-Nmatrix $\Mt=\tuple{A,\cdot_\Mt,D}$, we say that $a,b \in A$ are separated, written $a\# b$ if $a \in D$ and $b \notin D$, or vice-versa. A pair of sets of 
elements $X,Y \subseteq A$ are separated, written $X\# Y$, if for every $a \in X$ and $b \in Y$ we have that $a\#b$. Note that $X\#Y$ precisely if $X \subseteq D$ and
$Y \subseteq A\setminus D$, or vice versa.
A monadic formula $\varphi\in L_{\Sigma}(\{p\})$ such that $\varphi_{\Mt}(a)\#\varphi_{\Mt}(b)$ is said to separate $a$ and $b$.
We say that a set of monadic formulas $\mathsf{S}$ is a set of monadic separators for $\Mt$ when, for every pair of distinct truth-values of $\Mt$, there is a formula of 
$\mathsf{S}$ separating them.
An Nmatrix $\Mt$ satisfies monadicity (or simply, is monadic) if there is a set of monadic separators for $\Mt$. 
\begin{example}\label{ex:luk}
    Consider the signature $\Sigma = \{\neg, \vee, \wedge, \rightarrow\}$ and the $\Sigma$-matrix $\Mt_{\text{\L}} = \tuple{A,\cdot_{\text{\L}},D}$,
    with $A = \{0,\frac{1}{2},1\}$ and $D = \{1\}$, corresponding to {\L}ukasiewicz $3$-valued logic, with interpretations as described in the following tables.
    \begin{center}
        \begin{tabular}{c | c}
            $x$ & $\neg_{\text{\L}}(x)$ \\
            \hline
            $0$           & $1$ \\
            $\frac{1}{2}$ & $\frac{1}{2}$ \\
            $1$           & $0$ \\
        \end{tabular}
        \quad
        \begin{tabular}{c | c c c}
            $\vee_{\text{\L}}$   & $0$ & $\frac{1}{2}$ & $1$ \\
            \hline
            $0$           & $0$ & $\frac{1}{2}$ & $1$ \\
            $\frac{1}{2}$ & $\frac{1}{2}$ & $\frac{1}{2}$ & $1$ \\
            $1$           & $1$ & $1$ & $1$ \\
        \end{tabular}
        \quad
        \begin{tabular}{c | c c c}
            $\wedge_{\text{\L}}$   & $0$ & $\frac{1}{2}$ & $1$ \\
            \hline
            $0$           & $0$ & $0$ & $0$ \\
            $\frac{1}{2}$ & $0$ & $\frac{1}{2}$ & $\frac{1}{2}$ \\
            $1$           & $0$ & $\frac{1}{2}$ & $1$ \\
        \end{tabular}
        \quad
        \begin{tabular}{c | c c c}
            $\rightarrow_{\text{\L}}$   & $0$ & $\frac{1}{2}$ & $1$ \\
            \hline
            $0$           & $1$ & $1$ & $1$ \\
            $\frac{1}{2}$ & $\frac{1}{2}$ & $1$ & $1$ \\
            $1$           & $0$ & $\frac{1}{2}$ & $1$ \\
        \end{tabular}
    \end{center}
    $\LMt$ is monadic with $
    \{p, \neg p\}$ as a set of separators. Indeed, $p$ separates $1$ from $0$, and also $1$ from $\frac{1}{2}$, whereas $\neg p$ separates $0$ and $\frac{1}{2}$. One may 
    wonder, though, whether one could separate $0$ and $\frac{1}{2}$ without using negation.
    \hfill$\triangle$
\end{example}

\begin{remark}\label{remark:MonadicityIsDecidableForMatrices}
Notice that we can decide if a given matrix $\Mt$ is monadic by algorithmically generating every unary function expressible in $\Mt$, as it is usually done when calculating 
clones over finite algebras \cite{Lau}. This procedure is, however, quite expensive, since there are, at most, $n^n$ unary formulas from a set with $n$ values to itself.
The next example illustrates this procedure.
\end{remark}

\begin{example}\label{ex:luk2}
    Let $\LMt'$ be the $\{\vee, \wedge, \rightarrow\}$-reduct of $\LMt$ introduced in Example~\ref{ex:luk}, corresponding to the negationless fragment of {\L}ukasiewicz 
    $3$-valued logic. Let us show that $\LMt'$ is not monadic, by generating every unary expressible function.
    For simplicity, we represent a unary function $h : A \rightarrow A$ as a tuple $(h(0),h(\frac{1}{2}),h(1))$. 
    
    The formulas $p$, $p \vee p$ and $p \wedge p$ define the same function $(0, \frac{1}{2}, 1)$ and the formula $p \rightarrow p$ defines the constant function $(1,1,1)$. 
    It is easy to see that we cannot obtain new functions by further applying connectives, so $(0,\frac{1}{2},1)$ and $(1,1,1)$ are the only expressible unary functions.
    For example, $((p\rightarrow p)\rightarrow(p\rightarrow p))_{\LMt'} = (1,1,1)$. We conclude that $0$ cannot be separated from $\frac{1}{2}$, and so $\LMt'$ is not monadic.
    \hfill$\triangle$
\end{example}

\paragraph*{Monadicity in Nmatrices}
In the presence of non-determinism, we cannot follow the strategy described in Remark~\ref{remark:MonadicityIsDecidableForMatrices} and Example~\ref{ex:luk2}. A fundamental difference from the deterministic case is that the multi-functions represented 
by formulas in a Nmatrix are sensitive to the syntax since, even if there are multiple choices for the value of a subformula, all its occurrences need to 
have the same value. Crucially, on an Nmatrix $\Mt$, the multi-function  $\conn(\varphi_1,\ldots,\varphi_k)_\Mt$ does not depend only on the multi-functions $\conn_\Mt$ and
$(\varphi_1)_\Mt$, \ldots, $(\varphi_k)_\Mt$, as we shall see in the next example.

Hence, contrarily to what happens in the deterministic case, when generating the expressible multi-functions in an Nmatrix $\Mt$ (to find if $\Mt$ is monadic, or for any 
other purpose), we cannot just keep the information about the multi-functions themselves but also about the formulas that produce them. Otherwise, we might generate a 
non-expressible function (as every occurrence of a subformula must have the same value) or miss some multi-functions that are still expressible.


With the intent of making the notation lighter, when representing the interpretation of the connectives using tables, we drop the curly brackets around the elements of an output 
set. For example, in the table of Example~\ref{example:SubformulaDependenceOnExpressedFunctions} we simply write $a,c$ instead of $\{a,c\}$.

\begin{example}\label{example:SubformulaDependenceOnExpressedFunctions}
    Consider the signature $\Sigma$ with only one binary connective $g$, and two Nmatrices $\Mt=\tuple{\{a,b,c\},\cdot_{\Mt},D}$ and
    $\Mt'=\tuple{\{a,b,c\},\cdot_{\Mt'},\{c\}}$, with interpretation described in the following 
    table.
    \begin{center}
        \begin{tabular}{c | c c c }
            $g_\Mt$ & $a$& $b$ & $c$  \\
            \hline
            $a$&  $c$ &$ a $ & $ b,c $  \\
            $b$&  $b$ &$c $ & $ a,c $  \\
            $c$ &$ b,c$  &$a,c$ & $ c$  
        \end{tabular}
        \qquad\qquad
        \begin{tabular}{c | c c c }
            $g_{\Mt'}$ & $a$& $b$ & $c$  \\
            \hline
            $a$&  $c$ &$ a $ & $ c $  \\
            $b$&  $b$ &$c $ & $ a $  \\
            $c$ &$ b,c$  &$a,c$ & $ c$  
        \end{tabular}
    \end{center}
    Let $\varphi=g(g(p,p),p)$ and $\psi=g(p,g(p,p))$. In $\Mt$, $\varphi_{\Mt} = \psi_{\Mt} = (\{b,c\},\{a,c\},\{c\})$.
    Although these formulas define the same multi-function, they are not interchangeable, as $g(\varphi,\varphi)_\Mt=g(\psi,\psi)_\Mt=g(p,p)_\Mt = (\{c\},\{c\},\{c\})$ but
    $g(\varphi,\psi)_\Mt=g(\psi,\varphi)_\Mt=(\{a,c\},\{b,c\},\{c\})$, thus illustrating the already mentioned sensitivity to the syntax.
    Still, consider $v_x : L_{\Sigma}(P) \to \{a,b,c\}$, with $x \in \{a,b,c\}$, defined as
    \begin{equation*}
        v_x(\gamma) =
        \begin{cases}
            x &\text{if } \gamma \in P \\
            c &\text{otherwise}.
        \end{cases}
    \end{equation*}
    These functions can easily be shown to be valuations over $\Mt$ for every choice of $x$ and so, for every unary formula $\varphi \neq p$, we have that $c \in \varphi_{\Mt}(a) 
    \cap \varphi_{\Mt}(b) \cap \varphi_{\Mt}(c)$. We conclude that, apart from $p$, no unary formula can separate elements of $\{a,b,c\}$, and so $\Mt$ is not monadic, for 
    any choice of $D$.

    In $\Mt'$ the outcome is radically different. As $g(p,p)_{\Mt'}(x) = \{c\}$ for 
    every $x \in \{a,b,c\}$, we have $g(p,g(p,p))_{\Mt'}(a) = \{c\}$ and $g(p,g(p,p))_{\Mt'}(b) = \{a\}$ and so, in this case, $\Mt'$ is monadic with set of separators
    $\{p,g(p,g(p,p))\}$.
    \hfill$\triangle$
\end{example}



\section{Counter machines and Nmatrices}\label{sec:machines}
In this section we recall the essentials of counter machines, and define a suitable finite Nmatrix representing the computations of any given counter machine.
\subsection*{Counter machines}
A \emph{(deterministic) counter machine} is a tuple $\mC = \langle n, Q, \qi, \delta \rangle$, where $n \in \mathbb{N}$ is the number of counters, $Q$ is a finite set of 
states, $\qi \in Q$ is the initial state, and $\delta$ is a partial transition function $\delta:Q\not\to  (\{\inc{i} : 1 \leq i \leq n\} \times Q) \cup (\{\test{i} : 1 \leq 
i \leq n\} \times Q^2)$.
    
The set of halting states of $\mC$ is denoted by  $H=\{q\in Q:\delta(q)\text{ is undefined}\}$. 

A configuration of $\mC$ is a tuple $C=(q,\vec{x}) \in Q \times \mathbb{N}_0^n$, where $q$ is a state, and $\vec{x}=x_1,\ldots,x_n$ are the values of the counters. 
Let $\mathsf{Conf}(\mC)$ be the set of all configurations.
$C\in \mathsf{Conf}(\mC)$ is said to be the initial configuration if $q=\qi$ and $\vec{x}=\vec{0}$. $C$ is said to be a halting configuration if $q \in H$. 

When $(q,\vec{y})$ is not a halting configuration, the transition function $\delta$ completely determines the next configuration $\textsf{nxt}(q,\vec{y})$ as follows:
\begin{equation*}
    \nxt(q, \vec{y}) =
    \begin{cases}
        (q', \vec{y} + \vec{\mathsf{e}_i}) &\text{if } \delta(q) = \tuple{\inc{i}, q'}, \\
        (q'', \vec{y}) &\text{if } \delta(q) = \tuple{\test{i}, q'', q'''} \text{ and } x_i = 0, \\
        (q''', \vec{y} - \vec{\mathsf{e}_i}) &\text{if } \delta(q) = \tuple{\test{i}, q'', q'''} \text{ and } x_i \neq 0,
    \end{cases}
\end{equation*}
where $\vec{\mathsf{e}_i}$ is such that $(\mathsf{e}_i)_i = 1$ and $(\mathsf{e}_i)_j = 0$, for all $j \neq i$.

The computation of $\mC$ is a finite or infinite sequence of configurations $\tuple{C_i}_{i<\eta}$, where $\eta\in\mathbb{N}_0\cup\{\omega\}$ such that $C_0$ is the initial 
configuration, and for each $i<\eta$, either $C_i$ is a halting configuration and $i+1=\eta$ is the length of the computation, or else $C_{i+1}=\textsf{nxt}(C_i)$. \pagebreak

The intuition behind the transitions of a counter machine is clear from the underlying notion of computation, and in particular the definition of the next configuration. 
Clearly, $\delta(q)=\tuple{\inc{i},r}$ results in incrementing the $i$-th counter and moving to state $r$, whereas $\delta(q)=\tuple{\test{i},r,s}$ either moves to state $r$,
leaving the counters unchanged, when the value of the $i$-th counter is zero, or moves to state $s$, and decrements the $i$-th counter, when its value is not zero. As usual in 
counter machine models (see~\cite{minsky}), and also for the sake of simplicity, we are assuming that in the initial configuration all counters have value 0. This is well 
known not to hinder the computational power of the model, as a machine can always start by setting the counters to other desired input values. We will base our undecidability 
result on the following well known result about counter machines.

\begin{theorem}[\cite{minsky}]\label{thm:Halt}
    It is undecidable if a given counter machine halts when starting with all counters set to zero.
\end{theorem}

In what follows, given $\vec{x} \in A^n$ and $f : A \rightarrow B$, we define $f(\vec{x}) \in B^n$ as $f(\vec{x}) = \tuple{f(x_i) : 1 \leq i \leq n}$.
For a given counter machine $\mC = \langle n, Q, \qi, \delta \rangle$, we define the signature $\Sigma_{\mC}$ such that $\Sigma_{\mC}^{(0)} = \{\cz, \cep\}$,
$\Sigma_{\mC}^{(1)} = \{\suc\}$, $\Sigma_{\mC}^{(n+1)} = \{\step_q : q \in Q\}$ and $\Sigma_{\mC}^{(j)} = \emptyset$, for all $j \notin \{0,1,n+1\}$.
As usual, $\cz$ and $\suc$ allow us to encode every $k \in \mathbb{N}_0$ as the closed formula $\enc(k)=\suc^k(\cz)$. 
Moreover, we can encode every finite sequence of configurations $\tuple{C_0,\ldots,C_k}$ as a sequence formula in the following way:
\begin{itemize}
    \item $\seq(\tuple{}) = \epsilon$,  and $\seq(\tuple{C_0,\ldots,C_{k-1},(q, \vec{y})}) = \step_q(\seq(\tuple{C_0,\ldots,C_{k-1}}),\enc(\vec{y}))$.
\end{itemize}


We will construct a finite Nmatrix $\Mt_\mC$ over $\Sigma_\mC$ that recognizes as a theorem precisely the finite computation of $\mC$, if it exists.
This means that $\Mt_\mC$ can only falsify a formula $\varphi$ if it is not a sequence formula, or if $\varphi = \seq(\tuple{C_0,\ldots,C_k})$ but 
$C_0$ is not the initial configuration of $\mC$, or $C_k$ is not a halting configuration of $\mC$, or $\nxt(C_i)\neq C_{i+1}$ for some $0 \leq i < k$.\\[-8mm]


\subsection*{From counter machines to Nmatrices}
For a given counter machine $\mathcal{C} = \tuple{n, Q, \qi, \delta}$ let
\begin{equation*}
    \Rm=\{\n_{= 0},\n_{\geq 0},\n_{\geq 1},\n_{\geq 2}\} \qquad\text{ and }\qquad
    \Conf=\{\mathtt{conf}_{q,\overrightarrow{\n}} :q \in Q,\overrightarrow{\n}\in \Rm^n\}
\end{equation*}
and consider the Nmatrix $\Mt_\mathcal{C} = \tuple{A, \cdot_{\Mt_\mathcal{C}}, D}$ over $\Sigma_\mathcal{C}$, where 
\begin{align*}
    A =  \Rm \cup \Conf \cup  \{\mathtt{init},\mathtt{error}\} \qquad\text{ and }\qquad
    D= \{\mathtt{conf}_{q,\overrightarrow{\n}} :q \in H,\overrightarrow{\n}\in \Rm^n\}\qquad\text{ and }
\end{align*}
\begin{equation*}
    \zero_{\Mt_\mathcal{C}}  = \{\n_{= 0},\n_{\geq 0}\} \qquad \epsilon_{\Mt_\mC} = \{\mathtt{init}\} \qquad
    \suc_{\Mt_\mathcal{C}}(x) =
    \begin{cases}
        \{\n_{\geq 1}\} &\quad\text{if } x = \n_{= 0} \\
        \{\n_{\geq 0},\n_{\geq 1}\} &\quad\text{if } x = \n_{\geq 0} \\
        \{\n_{\geq 2}\} &\quad\text{if } x \in\{ \n_{\geq 1},\n_{\geq 2}  \}\\
        \{\mathtt{error}\} &\quad\text{otherwise}
    \end{cases}
\end{equation*}
\begin{equation*}
    (\step_q)_{\Mt_\mathcal{C}}(x,\vec{z}) =
    \begin{cases}
        \{\mathtt{conf}_{q,\vec{z}}\}
        &\text{if } x=\mathtt{init},q=q_{\mathsf{init}}\text{ and }\vec{z}\in \{\n_{= 0}\}^n\cup\{\n_{\geq 0}\}^n, \text{ or} \\
        &\text{if } x=\mathtt{conf}_{q',\overrightarrow{y}},  \overrightarrow{z} \in \Rm^n,  \text{ and }\\
        &\,\,\delta(q')=\tuple{\test{i},q,s},y_i\in\{\n_{= 0},\n_{\geq 0}\}\text{ and }\overrightarrow{y}=\overrightarrow{z},\text{ or }\\
        &\,\,\delta(q')=\tuple{\test{i},s,q},y_i\in \mathsf{suc}_{\Mt_\mathcal{C}}(z_i) \text{ and } z_j=y_j  \text{ for }j\neq i , \text{ or } \\
        &\,\,\delta(q')=\tuple{\inc{i},q},z_i\in \mathsf{suc}_{\Mt_\mathcal{C}}(y_i) \text{ and } z_j=y_j  \text{ for }j\neq i \\
        \{\mathtt{error}\} &\text{otherwise}
        \end{cases}
\end{equation*}
where, $s \in Q$ represents an arbitrary state.\pagebreak

$\Mt_\mC$ is conceived as a finite way of representing the behavior of $\mC$. For that purpose, it is useful to understand the operations 
$\mathsf{zero}$ and $\mathsf{suc}$ as means of representing the natural number values of the counters. Their interpretation, however, is finitely defined over the abstract 
values $\Num$. In fact, in order to check if some formula $\varphi$ encodes a sequence of computations respecting $\nxt$, it is not essential to distinguish all natural 
values. Indeed, it is easy to conclude from 
the definition of counter machine that in each computation step its counters either retain their previous values, or else they are incremented or decremented. As we set the 
initial configuration with all counters set to zero and the effect of test transitions also depends on detecting zero values, it is sufficient to being able to characterize 
unambiguously the value $0$ and, additionally, being able to recognize pairs of values whose difference is larger than one. This is successfully accomplished with the 
proposed non-deterministic interpretation of $\mathsf{suc}$, as shall be made clear below.
The $\epsilon$ and $\mathsf{step}$ operations are then meant to represent sequences of configurations, whereas their interpretation over the abstract values
$\Conf\cup\{\init\}$ guarantees that consecutive configurations respect $\nxt$. Of course, the designated values of $\Mt_\mC$ are those corresponding to halting 
configurations. The $\mathtt{error}$ value is absorbing with respect to the interpretation of all operations, and gathers all meaningless situations. Overall, as we will 
show, $\Mt_{\mC}$ induces a logic that has at most one theorem, corresponding to the computation of $\mC$, if it is halting.

\subsection*{The inner workings of the construction}
In the next examples we will illustrate the way the Nmatrix $\Mt_\mC$ encodes the computations of $\mC$. Proofs of the general statements are postponed to the next section.
\smallskip

In order for $\sder_{\Mt_\mC}$ to have, at most, the formula representing the computation of $\mC$ as theorem, $\Val(\Mt_\mC)$ must contain enough valuations to refute
every formula not representing the computation of $\mC$. These valuations are presented in the following example
\begin{example}\label{ex:v}
By definition of $\seq$, it is clear that no formula containing variables corresponds to a sequence of configurations. Furthermore, no formula containing
variables can be a theorem of $\Mt_\mC$ since these formulas are easily refuted by any valuation sending the variables to the truth value $\error$, as 
this value is absorbing (aka infectious), that is, $\suc_{\Mt_{\mC}}(x) = \error$ whenever $x = \error$ and $(\step_q)_{\Mt_\mC}(x,\vec{z}) = \error$ 
whenever $x = \error$ or $z_i = \error$. Because of this, from here onwards, we concern ourselves only with the truth-values assigned to closed formulas.

We do not have much freedom left, but it will be enough.
The interpretations of the connectives are all deterministic, except in the case of $\zero_{\Mt_\mC}$ and $\suc_{\Mt_\mC}(\n_{\geq 0})$. This means that,
if $v \in \Val(\Mt_\mC)$ and $v(\zero) = \n_{=0}$, then there is no choice left for the values assigned by $v$ to
the remaining closed formulas. Consider, therefore, the following valuation
\begin{equation*}
    v^=_0(\psi) =
    \begin{cases}
        \n_{=0} &\text{if } \psi = \zero, \\
        \n_{\geq 1} &\text{if } \psi = \enc(1), \\
        \n_{\geq 2} &\text{if } \psi = \enc(j) \text{ with } j \geq 2, \\
        \mathtt{init} &\text{if } \psi = \epsilon, \\
        (\step_q)_{\Mt_\mC}(v^=_0(\varphi),\vec{z}) &\text{if } \psi = \step_q(\varphi,\psi_1,\dots,\psi_n) \text{ and } z_i = v^=_0(\psi_i), \\
        \mathtt{error}, &\text{otherwise}.
    \end{cases}
\end{equation*}
If $v(\zero) = \n_{\geq 0}$, however, we can still loop the truth values assigned to formulas of the form $\enc(j)$. The amount of loops could be infinite or finite, though
the infinite case is of no interest to us, since it does not allow us to falsify any of the formulas that we want to falsify.

Let $k \in\mathbb{N}_0$, consider the valuations
\begin{equation*}
    v_k(\psi) =
    \begin{cases}
        \n_{\geq 0} &\text{if } \psi = \enc(j) \text{ with } j \leq k, \\
        \n_{\geq 1} &\text{if } \psi = \enc(j) \text{ with } j = k + 1, \\
        \n_{\geq 2} &\text{if } \psi = \enc(j) \text{ with } j \geq k + 2, \\
        \mathtt{init} &\text{if } \psi = \epsilon, \\
        (\step_q)_{\Mt_\mC}(v_k(\varphi),\vec{z}) &\text{if } \psi = \step_q(\varphi,\psi_1,\dots,\psi_n) \text{ and } z_i = v_k(\psi_i), \\
        \mathtt{error}, &\text{otherwise}.
    \end{cases}
\end{equation*}\\[-15mm]

\hfill$\triangle$
\end{example}

{}\smallskip As previously discussed, it is crucial for the valuations in $\Val(\Mt_\mC)$ to be able to identify whenever two given numbers $a,b \in \mathbb{N}_0$
are not consecutive, or different. To this aim, for every pair $a,b \in (\mathbb{N}_0)^2$, we denote by $\mu^{+}_{a,b},\mu^{-}_{a,b},\mu^{\neq}_{a,b}$ the valuations
determined by the following conditions
\begin{equation*}
    \mu^{+}_{a,b} =
    \begin{cases}
        v_a &\text{if } b \geq a + 1, \\
        v_{a-1} &\text{if } b \leq a \text{ and } a \neq 0, \\
        v^=_0 &\text{if } b \leq a \text{ and } a = 0.
    \end{cases}
    \quad
    \mu^{-}_{a,b} =
    \begin{cases}
        v_{a-2} &\text{if } b \leq a - 2, \\
        v_{a-1} &\text{if } b \geq a - 1 \text{ and } a \neq 0, \\
        v^=_0 &\text{if } b \geq a - 1 \text{ and } a = 0.
    \end{cases}
    \quad
    \mu^{\neq}_{a,b} =
    \begin{cases}
        v^=_0 &\text{if } a = 0, \\
        v_{a-1} &\text{if } a \neq 0.
    \end{cases}
\end{equation*}

\begin{remark}\label{remark:MuProperties}
    The following properties can be easily checked by inspecting the corresponding definition:
    \begin{itemize}
    \item if $b \neq a + 1$ then $\mu^{+}_{a,b}(\enc(b)) \notin \suc_{\Mt_\mC}(\mu^{+}_{a,b}(\enc(a)))$,
    \item if $b \neq a - 1$ then $\mu^{-}_{a,b}(\enc(a)) \notin \suc_{\Mt_\mC}(\mu^{-}_{a,b}(\enc(b)))$, and
    \item if $b \neq a$ then $\mu^{\neq}_{a,b}(\enc(b)) \neq \mu^{\neq}_{a,b}(\enc(a))$.
    \end{itemize}
\end{remark}\smallskip




In the following two examples we consider two different machines that should make clear the soundness of our construction. In the first one, we show how
every valuation validates the formula encoding a finite computation. We also see how sequences of configurations can fail to respect $\nxt$ in different ways
and how we can use the valuations presented in Example~\ref{ex:v} to falsify formulas encoding them. In the second example, we show a counter machine that never
halts.

\begin{example}
    Consider the counter machine $\mC = \tuple{1,Q,\qi,\delta}$ with $Q = \{\qi,q_1,q_2,q_3\}$ and $\delta$ as defined in the following table
    \begin{center}
        \begin{tabular}{ c | c c c c}
            $q$ & $\qi$ & $q_1$ & $q_2$ & $q_3$ \\
            \hline
            $\delta(q)$ & $\tuple{\inc{1},q_1}$ & $\tuple{\test{1},q_3,q_2}$ & $\tuple{\test{1},q_3,q_3}$ & undefined
        \end{tabular}
    \end{center}
    
    The only halting state of $\mC$ is $q_3$ and the machine $\mC$ has the following finite computation
    \begin{equation*}
        \tuple{(\qi,0),(q_1,1),(q_2,0),(q_3,0)}.
    \end{equation*}
    For every $v \in \Val(\Mt_\mC)$, we have $v(\cep)=\init$. The values of $v(\enc(k))$ are dependent on $v$:
    if $v = v^=_0$ then $v(\enc(0)) = \n_{=0}$ and $v(\enc(1)) = \n_{\geq 1}$.
    If $v \neq v^=_0$ then $v(\enc(0)) = \n_{\geq 0}$ and $v(\enc(1)) \in (\suc)_{\Mt_\mC}(\n_{\geq 0}) = \{\n_{\geq 0},\n_{\geq 1}\}$.
    
    Let $\varphi_j$ be the formula representing the prefix with only the first $j+1$ configurations, we obtain, from the above equalities
    and the definition of $\step_\Mt$, that
    \begin{align*}
        v(\varphi_0) &= v(\step_{\qi}(\epsilon,\enc(0))) = (\step_{\qi})_\Mt(\init,v(\enc(0)))= \conf_{\qi,v(\enc(0))} \\
        v(\varphi_1) &= v(\step_{q_1}(\varphi_0,\enc(1))) = (\step_{q_1})_\Mt(\conf_{\qi,v(\enc(0))},v(\enc(1)))= \conf_{q_1,v(\enc(1))} \\
        v(\varphi_2) &= v(\step_{q_2}(\varphi_1,\enc(0))) = (\step_{q_2})_\Mt(\conf_{q_1,v(\enc(1))},v(\enc(0)))= \conf_{q_2,v(\enc(0))} \\
        v(\varphi_3) &= v(\step_{q_3}(\varphi_2,\enc(0))) = (\step_{q_3})_\Mt(\conf_{q_2,v(\enc(0))},v(\enc(0)))= \conf_{q_3,v(\enc(0))}
    \end{align*}

    The formula $\varphi_3$ encodes the finite computation of $\mC$ and, since $\conf_{q_3,v(\enc(0))} \in D$, $\emptyset \sder_{\Mt_\mC}
    \varphi_3$. Furthermore, the formulas $\varphi_i$ with $0\leq i\leq 3$, that encode its strict prefixes, are falsified by all valuations, since
    $q_3$ is the only halting state.
    
    Formulas not representing sequences of configurations, like $\suc(\psi)$ with
    $\psi\neq \enc(j)$, are falsified by every $v \in \Val(\Mt_\mC)$ since $v(\suc(\psi))=\suc_{\Mt}=\error$.
    Formulas encoding sequences of configurations not starting in the initial configuration of $\Mt$ are also falsifiable:
    $v_0(\step_q(\psi,\enc(j)))=\error$ whenever, either $q\neq \qi$ and $\psi=\cep$, or,
    $q=\qi$, $\psi=\cep$ and $j\neq 0$. For example, 
    $v_0(\step_{\qi}(\cep,\enc(1)))=(\step_{\qi})_{\Mt_{\mC}}(\init,\n_{\geq 1}) =\error$.
    
    The sequence $\tuple{(\qi,0),(q_1,2)}$, encoded by $\psi=\step_{q_1}(\varphi_0,\enc(2))$, illustrates a situation
    where the value in the counter was incremented by two while
    the transition $\delta(q_1)=\tuple{\inc{1},q_2}$ required it to increase by only one. In this case, we have
    $\mu^{+}_{0,2}(\psi)=v^=_0(\psi)=(\step_{q_1})_\Mt(\conf_{\qi,\n_{= 0}},\n_{\geq 2})=\error$.
    In the same way, we also have $\mu^{-}_{2,0}(\gamma)=\error$ with $\gamma=\step_{q_2}(\varphi_1,\enc(2))$.
    This reflects the fact that $\gamma$ encodes the sequence resulting from appending $(q_2,2)$ to the sequence encoded by $\varphi_1$, hence
    incrementing the value the counter while $\delta(q_1)$ required it to be decremented by one.

    Finally, consider $\xi=\step_{q_3}(\varphi_2,\enc(1))$ encoding the sequence resulting from appending $(q_3,1)$ to the sequence encoded by 
    $\varphi_2$. As the value in the first counter was incremented, while the transition $\delta(q_2)=\tuple{\test{1},q_3,q_3}$ required it to remain 
    unchanged we obtain $\mu^{\neq}_{0,1}(\xi) = v^=_0(\xi) = \error$.
    \hfill$\triangle$
\end{example}

\begin{example}
    Consider the counter machine $\mC = \tuple{2,Q,\qi,\delta}$ with $Q = \{\qi,q_1,q_2,q_3,q_4\}$ and $\delta$ as defined in following table
    \begin{center}
        \begin{tabular}{c | c c c c c}
            $q$ & $\qi$ & $q_1$ & $q_2$ & $q_3$ & $q_4$ \\
            \hline
            $\delta(q)$ & $\tuple{\inc{1},q_1}$ & $\tuple{\inc{2},q_2}$ & $\tuple{\test{1},q_4,q_3}$ & $\tuple{\inc{1},\qi}$ & undefined
        \end{tabular}
    \end{center}
    This machine does not have a finite computation, and its infinite computation loops indefinitely in the following cycle consisting of $4$ transitions.
    It starts by incrementing both counters.
    Then it tests if the first counter has the value $0$.
    As the counter has just been incremented, the test is bound to fail and hence that counter is decremented.
    It then increments the same counter, and returns to the initial state.
  
    The halting state $q_4$ could only be reached if at a certain point the test would succeed, but this never happens since, at the point the tests are 
    made, the value of the first counter is never $0$. Thus, the machine $\mC$ has the following infinite computation
    \begin{equation*}
        \tuple{(\qi,0,0),(q_1,1,0),(q_2,1,1),(q_3,0,1),(\qi,1,1),(q_1,2,1),(q_2,2,2),(q_3,1,2),(\qi,2,2),\ldots}
    \end{equation*}
    As we show in the next section, since $\mC$ has no finite computations, $\Mt_{\mC}$ has no theorems.
    Let $k\geq 1$ and consider the sequence resulting from adding $(q_3,k-1,k+1)$ to the prefix of the infinite computation of $\mC$ with 
    $4(k-1)+3$ elements, and let $\varphi_k$ encode this sequence. In this case, the value of the second counter is increased, when it should have 
    remained the same, and we have
    \begin{equation*}
        \mu^{\neq}_{k,k+1}(\varphi_k) = v_{k-1}(\varphi_k) = (\step_{q_3})_{\Mt_\mC}(\conf_{q_2,\n_{\geq 1},\n_{\geq 1}},\n_{\geq 0},\n_{\geq 2})
        = \error.
    \end{equation*}\vspace{-3\baselineskip}

\hfill $\triangle$
    
\end{example}

\section{Monadicity of Nmatrices is undecidable}\label{sec:undecided}
In this section we show that $\Mt_\mC$ really does what is intended. The main result of the paper then follows, after we additionally introduce a 
construction connecting the existence of a theorem with monadicity.

\subsection*{$\Mt_\mC$ validates the finite computation of $\mC$}
In the following propositions we show that $\Mt_\mC$ is interpreting computations of $\mC$ as it should.
Thus, formulas encoding computations of $\mC$ that do not end in a halting state can be falsified, whilst
the one encoding the finite computation of $\mC$ is always designated.

\begin{proposition}\label{prop:StepValidatesNxt}
    Let $\mC = \tuple{n,Q,\qi,\delta}$ be a deterministic $n$-counter machine. If $\nxt(q,\vec{y}) = (q',\vec{z})$ then
    \begin{equation}\label{equation:ValidStep}
        (\step_{q'})_{\Mt_\mC}(\conf_{q,v(\enc(\vec{y}))}, v(\enc(\vec{z}))) = \{\conf_{q',v(\enc(\vec{z}))}\}, \text{ for every } v \in \Val(\Mt_\mC).
    \end{equation}
\end{proposition}
\begin{proof}
    Suppose that $\nxt(q,\vec{y}) = (q',\vec{z})$. We have to consider three cases, depending on 
     $\delta(q)$. 
    
    If $\delta(q) = (\inc{i},q')$ and $\vec{z} = \vec{y} + \vec{\mathsf{e}_i}$ then, for every $j \neq i$ and $v \in \Val(\Mt_\mC)$, we have
    $z_j = y_j$ and $v(\enc(z_j)) = v(\enc(y_j))$. Furthermore, $z_i = y_i + 1$, so $\enc(z_i) = \suc(\enc(y_i))$ and, for every $v \in \Val(\Mt_\mC)$,
    $v(\enc(z_i)) = v(\suc(\enc(y_i))) \in \suc_{\Mt_\mC}(v(\enc(y_i)))$. 

 Otherwise, $\delta(q) = \tuple{\test{i}, s_1, s_2}$ for some machine states $s_1$ and $s_2$.

    If $s_1=q'$,
     $y_i = 0$ and $\vec{z} = \vec{y}$. Then, for every
    $v \in \Val(\Mt_\mC)$, we have $v(\enc(\vec{z})) = v(\enc(\vec{y}))$ and $v(\enc(z_i)) = v(\enc(y_i)) = v(\zero) \in \{\n_{=0}, \n_{\geq 0}\}$.

    If $s_2=q'$, 
    $y_i \neq 0$ and $\vec{y} = \vec{z} + \vec{\mathsf{e}_i}$.
    Then, for every $j \neq i$ and $v \in \Val(\Mt_\mC)$, we have $z_j = y_j$ and so $v(\enc(z_j)) = v(\enc(y_j))$. Furthermore, $y_i = z_i + 1$,
    so $\enc(y_i) = \suc(\enc(z_i))$ and, for every $v \in \Val(\Mt_\mC)$, $v(\enc(y_i)) = v(\suc(\enc(z_i))) \in \suc_{\Mt_\mC}(v(\enc(z_i)))$. 
  
  In all the three cases we conclude
    \eqref{equation:ValidStep} directly by the definition of $(\step_{q'})_{\Mt_\mC}$.
\end{proof}

\begin{theorem}
    \label{theorem:IfFiniteComputationThenTheorem}
    Let $\mC = \tuple{n,Q,\qi,\delta}$ be a deterministic $n$-counter machine with a finite computation $\tuple{C_0,\dots,C_k}$, then
    $\emptyset \sder_{\Mt_\mC} \seq(\tuple{C_0,\dots,C_k})$.
\end{theorem}
\begin{proof}
    Suppose that $\tuple{C_0,\dots,C_k}$ is the finite computation of $\mC$. Then we have that 
    $C_0 = (\qi, \vec{\zero})$, where $\vec{\zero} =
    \tuple{\zero,\dots,\zero}$, 
     $\nxt(C_j) = C_{j+1}$ for every $0 \leq j < k$ and 
     $C_k = (q_k, \vec{z})$ is a halting configuration.
    For every $v \in \Val(\Mt_\mC)$ we have $v(\seq(\tuple{C_0})) = \conf_{\qi,v(\vec{\zero})}$ and, by proposition~\ref{prop:StepValidatesNxt},
    $v(\tuple{C_0,\dots,C_k}) = \conf_{q_k,v(\enc(\vec{z}))}$, where $C_k = (q_k, \vec{z})$. Since $C_k$ is a halting configuration, we conclude that
    $\conf_{q_k,v(\enc(\vec{z}))} \in D$, and so $\emptyset \sder_{\Mt_\mC} \seq(\tuple{C_0,\dots,C_k})$.
\end{proof}

\subsection*{$\Mt_\mC$ can falsify everything else}
The following propositions deal with all the possible ways in which a formula can fail to represent a halting 
computation of $\mC$.

\begin{proposition}\label{prop:RefutesNonSequences}
    Let $\mC = \tuple{n,Q,\qi,\delta}$ be a deterministic $n$-counter machine.
    If $\varphi \in \fm{\Sigma_\mC}(\emptyset)$ does not represent a sequence of configurations of $\mC$ then $v(\varphi) 
    \neq \conf_{q,\vec{y}}$ for all $v\in \Val_{\mC}$, $q \in Q$ and $\vec{y} \in \Rm^n$.
\end{proposition}
\begin{proof}
    The proof follows by induction on the structure of the formula $\varphi\in \fm{\Sigma_\mC}(\emptyset)$.

    In the base case we have $\varphi \in \{\cz,\cep\}$. The statement then holds since, for all $v\in \Val_{\mC}$, $v(\varphi) \in \{\zero,\init\}$.

    For the step we have two cases.
    In the first case, $\varphi=\suc(\psi)$ and $v(\suc(\psi))\in \Rm \cup \{\error\}$. In the second case, $\varphi=\mathtt{step}_q(\psi,\psi_1,\ldots,
    \psi_n)$ and, if $\varphi$ does not represent any sequence of configurations, then one of the following must hold
    \begin{itemize}
        \item $\psi$ does not represent a sequence of configurations, in which case, by induction hypothesis, $v(\psi)\neq\conf_{q,\vec{y}}$, or
        \item $\psi_i \neq \enc(j)$, for some $1 \leq i \leq n$, so $v(\psi_i)\notin\Num$.
    \end{itemize}
    In both cases we have  $v(\varphi)=\suc_{\Mt_\mC}(v(\psi),v(\psi_1),\ldots,v(\psi_n))=\error$.
\end{proof}

\begin{proposition}
    \label{prop:NotRespectingNxtEntailsError}
    Given a deterministic counter machine $\mC = \tuple{n,Q,\qi,\delta}$. If $\nxt(q,\vec{y}) \neq (q',\vec{z})$ then
    \begin{equation}\label{equation:NotValidStep}
        (\step_{q'})_{\Mt_\mC}(\conf_{q,v(\enc(\vec{y}))}, v(\enc(\vec{z}))) = \{\error\},
    \end{equation}
    for some $v \in \Val(\Mt_\mC)$.
\end{proposition}
\begin{proof}

    Assume $\nxt(q,\vec{y}) \neq (q',\vec{z})$ and notice that, if $\delta(q)$ concerns the $i$th counter and $y_j \neq z_j$, for some $j \neq i$, 
    then $\mu^{\neq}_{y_j,z_j}(\enc(y_j)) \neq \mu^{\neq}_{y_j,z_j}(\enc(z_j))$, by remark~\ref{remark:MuProperties}. Therefore,
    equality~\eqref{equation:NotValidStep} holds for $v = \mu^{\neq}_{y_j,z_j}$. Because of this, throughout the rest of the proof, we assume
    that, $y_j = z_j$, for every $j \neq i$, and concern ourselves only with the values taken by $y_i$ and $z_i$.
%
    We have to consider three cases, depending on 
     $\delta(q)$.

    If $\delta(q) = \tuple{\inc{i}, s}$, for some machine state $s$. Then either
    $q' \neq s$, in which case equality~\eqref{equation:NotValidStep} holds for every $v \in \Val(\Mt_\mC)$, or $z_i \neq y_i + 1$.
    In this later case, also by remark~\ref{remark:MuProperties}, we have $\mu^{+}_{y_i,z_i}(\enc(z_i)) \notin \suc_{\Mt_\mC}(\mu^{+}_{y_i,z_i}
    (\enc(y_i)))$, so equality~\eqref{equation:NotValidStep} holds for $v = \mu^{+}_{y_i,z_i}$.

   Otherwise, $\delta(q) = \tuple{\test{i}, s_1, s_2}$ for some machine states $s_1$ and $s_2$. 
    
    If $y_i = 0$, then either $q' \neq s_1$ or $z_i \neq y_i$. Consider the valuation $\mu^{=}_{y_i,z_i} = v^=_0$. Since $v^=_0(y_i) = \n_{=0}
    \notin \suc_{\Mt_\mC}(v^=_0(\enc(z_i)))$, the second condition concerning $\test{i}$, in the definition of $(\step_{q'})_{\Mt_\mC}$, is not
    satisfied whenever $v = \mu^{=}_{y_i,z_i}$. The first condition is also not satisfied if $q' \neq s_1$, directly, or if $z_i \neq y_i$, since in this case
    $\mu^{=}_{y_i,z_i}(z_i) \neq \mu^{=}_{y_i,z_i}(y_i)$, by remark~\ref{remark:MuProperties}. We conclude that equality~\eqref{equation:NotValidStep} holds for $v = \mu^{=}_{y_i,z_i}$.
    
    If $y_i \neq 0$, then either $q' \neq s_2$ or $z_i \neq y_i - 1$. Note that, in any case, $\mu^{-}_{y_i,z_i}(\enc(y_i)) \notin
    \{\n_{=0},\n_{\geq 0}\}$, which can easily be checked using the definition of $\mu^{-}_{y_i,z_i}$. Therefore, the first condition
    concerning $\test{i}$, in the definition of $(\step_{q'})_{\Mt_\mC}$, is not satisfied whenever $v = \mu^{-}_{y_i,z_i}$.
    The second condition is also not satisfied if
    $q' \neq s_2$, directly, or if $z_i \neq y_i - 1$, since in this case $\mu^{-}_{y_i,z_i}(\enc(y_i)) \notin
    \suc_{\Mt_\mC}(\mu^{-}_{y_i,z_i}(\enc(z_i)))$, by remark~\ref{remark:MuProperties}. We conclude that equality~\eqref{equation:NotValidStep} 
    holds for $v = \mu^{-}_{y_i,z_i}$.
\end{proof}

\begin{proposition}\label{prop:RefutesNonComputations}
    Given a deterministic counter machine $\mC = \tuple{n,Q,\qi,\delta}$ and $\varphi \in \fm{\Sigma_\mC}(\emptyset)$ such that
    $\varphi = \seq(\tuple{C_0,\dots,C_k})$. If $\tuple{C_0,\dots,C_k}$ is not the computation of $\mC$ then $\emptyset \not\sder_{\Mt_\mC} \varphi$. 
\end{proposition}
\begin{proof}
    If $\tuple{C_0,\dots,C_k}$ is not the computation of $\mC$, then one of the following must hold: (i) $C_0$ is not the initial configuration of $\mC$,
    (ii) $C_k$ is not a halting configuration of $\mC$, or (iii) there is some $1 \leq i < k$ such that $\nxt(C_i) \neq C_{i+1}$. We deal with each
    situation separately.

    If (i) holds and $C_0 = (q,\vec{y})$ then either $q \neq \qi$ or $y_j \neq 0$, for some $1 \leq j \leq n$. In the first case, for all $v \in
    \Val(\Mt_\mC)$, we have $(\step_q)_{\Mt_\mC}(\init,v(\enc(\vec{y}))) = \error$. In the second case, $v_{y_j - 1}(\enc(y_j)) \notin \{\n_{=0},
    \n_{\geq 0}\}$ and $(\step_{\qi})_{\Mt_\mC}(\init,v_{y_j - 1}(\enc(\vec{y}))) = \error$.

    If (ii) holds and $C_k = (q,\vec{y})$ then $q$ is not a halting state and
    $v(\varphi) \in \{\error, \conf_{q,v(\vec{y})}\} \subseteq A \setminus D$, 
    for all $v \in \Val(\Mt_\mC)$.

    If (iii) holds then, by proposition~\ref{prop:NotRespectingNxtEntailsError}, there is $v \in \Val(\Mt_\mC)$ such that 
    $v(\seq(\tuple{C_0,\dots,C_{i+1}})) = \error$.

    In any of the cases, there is some $v \in \Val(\Mt_\mC)$ such that $v(\varphi) \notin D$, so $\emptyset\not\sder_{\Mt_\mC} \varphi$.
\end{proof}

Having seen how to refute any formula not representing a computation of $\mC$ we conclude $\Mt_\mC$ does exactly what we intended.

\begin{theorem}\label{thm:teohalt}
    Let $\mC = \tuple{n,Q,\qi,\delta}$ be a deterministic counter machine. For any formula $\varphi\in \fm{\Sigma_\mC}(P)$ we have
    $\emptyset \sder_{\Mt_\mC} \varphi$ if and only if $\varphi = \seq(\tuple{C_0,\dots,C_k})$ and $\tuple{C_0,\dots,C_k}$ is a finite
    computation of $\mC$.
\end{theorem}
\begin{proof}
    From right to left, if $\tuple{C_0,\dots,C_k}$ is a finite computation of $\mC$ and $\varphi = \seq(\tuple{C_0,\dots,C_k})$ then, by
    theorem~\ref{theorem:IfFiniteComputationThenTheorem}, we have that $\emptyset \sder_{\Mt_\mC} \varphi$. In the other direction, suppose
    $\emptyset \sder_{\Mt_\mC} \varphi$ then, as discussed in example~\ref{ex:v}, $\varphi$ must be a closed formula. By
    proposition~\ref{prop:RefutesNonSequences}, $\varphi = \seq(\tuple{C_0,\dots,C_k})$ for some sequence of configurations $\seq(\tuple{C_0,\dots,C_k})$,
    and, by proposition~\ref{prop:RefutesNonComputations}, $\seq(\tuple{C_0,\dots,C_k})$ must be the computation of $\mC$.
\end{proof}

\subsection*{From theoremhood to monadicity}
In order to obtain the announced undecidability result we need one last construction.
We will show how to build an Nmatrix $\Mt_{\mathsf{m}}$ from an Nmatrix $\Mt$, under certain conditions, such that $\Mt_{\mathsf{m}}$ is monadic if and 
only if $\sder_{\Mt}$ has theorems. 

Given a finite $\Sigma$-Nmatrix $\Mt=\tuple{A,\cdot_\Mt,D}$, let $\Sigma_{\mathsf{m}}$ be such that $\Sigma_{\mathsf{m}}^{(2)} = \Sigma^{(2)} \cup
\{f_a : a \in A \}$ and $\Sigma_{\mathsf{m}}^{(k)} = \Sigma^{(k)}$, for every $k \neq 2$.

Let $A_{\mathsf{m}}=A\cup \{1\}$, assuming w.l.g. that $1\notin A$, consider  $\Mt_{\mathsf{m}}=\tuple{A_{\mathsf{m}},\cdot_{\mathsf{m}},\{1\}}$ the 
$\Sigma_{\mathsf{m}}$-Nmatrix where, for each $g \in \Sigma^{(k)}$,
\begin{equation*}
    g_{\mathsf{m}}(x_1,\ldots,x_k)=
    \begin{cases}
        g_\Mt(x_1,\ldots,x_k)& \text{ if } x_1,\ldots,x_k\in A\\
        A_{\mathsf{m}}&\text{ otherwise }
    \end{cases}
\end{equation*}
and, for every $a \in A$,
\begin{equation*}
    (f_{a})_{\mathsf{m}}(x,y)=
    \begin{cases}
        \{1\}& \text{ if } x=a\text{ and }y\in D\\
        \; A& \text{ if }  x \in A\setminus \{a\}\text{ and }y\in D\\
        \; A_{\mathsf{m}}&\text{ otherwise }
    \end{cases}
\end{equation*}

The following theorem targets Nmatrices with infectious values. Recall that $*$ is infectious in $\Mt$ if for every connective $\conn$ in the 
signature of $\Mt$ we have $\conn_{\Mt}(x_1,\ldots,x_k)=*$ whenever $*\in \{x_1,\ldots,x_k\}$.

\begin{proposition}\label{prop:TheoremsIifMonadic}
    Given Nmatrix  $\Mt$ with at least two truth-values and among them an infectious non-designated value, $\sder_\Mt$ has theorems if and only if
    $\Mt_{\mathsf{m}}$ is monadic.
\end{proposition}
\begin{proof}
    Let us denote the infectious non-designated value of $\Mt$ by $*$.
    Clearly, $*$ ceases to be infectious in $\Mt_{\mathsf{m}}$ as $(f_a)_{\Mt_{\mathsf{m}}}$ does not necessarily output $*$ when it receives it as 
    input. 
    The value $1$ is also not infectious in $\Mt_{\mathsf{m}}$, quite the opposite, when given as input to any connective the output can take any value.
    That is, for every connective $\conn\in \Sigma_{\mathsf{m}}$ we have $\conn_{\Mt_{\mathsf{m}}}(x_1,\ldots,x_k)=A_\mathsf{m}$ whenever $1\in \{x_1,
    \ldots,x_k\}$. This immediately implies that if $\psi\in \Sub(\varphi)\setminus \{\varphi\}$ and $1\in \psi_{\Mt_\mathsf{m}}(x)$ then
    $\varphi_{\Mt_\mathsf{m}}(x)=A_\mathsf{m}$ for any $x\in A_\mathsf{m}$.
     
    If $\emptyset \sder_\Mt\varphi$ then $\varphi$ must be a closed formula due to the presence of $*$.
    Hence, for $v\in\Val(\Mt_{\mathsf{m}})$ we have $v(\varphi)\in D$.
    Thus, $\{p\}\cup \{f_a(p,\varphi):a\in A\}$ is a set of monadic separators for $\Mt_{\mathsf{m}}$, as $p$ separates $1$ from the elements in $A$, and
    $f_a(p,\varphi)$ separates $a$ from every $b\in A$.
      
    If instead there are no theorems in $\sder_\Mt$, let us consider an arbitrary monadic formula $\varphi\in L_{\Sigma_{\mathsf{m}}}(\{p\})$ and show it 
    cannot separate $*$ from the other elements of $A$.
    We need to consider two cases.
%
%
      \begin{itemize}
          \item  If $\varphi\in L_{\Sigma}(\{p\})$ then $\varphi_{\Mt_{\mathsf{m}}}(a)=\varphi_{\Mt}(a)\subseteq A\not\ni 1$ for every $a\in A$. 
              In which case $\varphi$ cannot separate any pair of distinct elements of $A$ and, in particular, cannot separate $*$ from any other
              element of $A$.
         \item If $\varphi\in L_{\Sigma_{\mathsf{m}}}(\{p\})\setminus L_{\Sigma}(\{p\})$ 
              then there is $f_a(\psi_1,\psi_2)\in \Sub(\varphi)$ with $\psi_1,\psi_2\in L_{\Sigma}(\{p\})$.
              If $p$ occurs in $\psi_2$ then $(\psi_2)_{\Mt_\mathsf{m}}(*)=(\psi_2)_{\Mt}(*)=\{*\}$ and
              $(f_a(\psi_1,\psi_2))_{\Mt_{\mathsf{m}}}(*)= A_{\mathsf{m}}$, since $* \notin D$.
             If $p$ does not occur in $\psi_2$ then, since $\emptyset\not\sder_\Mt\psi_2$, $(\psi_2)_{\Mt}\cap (A\setminus D)\neq \emptyset$ and 
              we also obtain $(f_a(\psi_1,\psi_2))_{\Mt_{\mathsf{m}}}(*)= A_{\mathsf{m}}$.
    
              Therefore, $\varphi_{\Mt_{\mathsf{m}}}(*)= A_{\mathsf{m}}$ since either $\varphi=f_a(\psi_1,\psi_2)$ or $f_a(\psi_1,\psi_2)\in \Sub(\varphi)
             \setminus \{\varphi\}$ and $1\in(f_a(\psi_1,\psi_2))_{\Mt_{\mathsf{m}}}(*)$. As $\varphi_{\Mt_{\mathsf{m}}}(*)$ contains both designated and 
            non-designated elements it cannot separate $*$ from any other element of $A$.
      \end{itemize}

    As we are assuming that $A$ has at least two elements, we conclude that $\Mt_{\mathsf{m}}$ is not monadic.
\end{proof}

Finally, we get to the main result of this paper.

\begin{theorem}
    The problem of determining if a given finite $\Sigma$-Nmatrix is monadic is undecidable.
\end{theorem}
\begin{proof}
    For every counter machine $\mC$, the Nmatrix $\Mt_{\mC}$ is in the conditions of Theorem~\ref{prop:TheoremsIifMonadic}, as it has more than two 
    truth-values and $\error$ is infectious.
    Therefore, by applying successively Theorem~\ref{thm:teohalt} and~\ref{prop:TheoremsIifMonadic}, we reduce the halting problem for counter machines
    to the problem of checking if a finite Nmatrix is monadic.
    Indeed, for a given counter machine $\mC$, $\mC$ halts if and only if $\sder_{\Mt_\mC}$ has theorems if and only if $(\Mt_\mC)_{\mathsf{m}}$ is monadic.
    Furthermore, the presented constructions are all computable and $(\Mt_\mC)_{\mathsf{m}}$ is always finite since, if $\mC$ has $m$ states and $n$ 
    counters, then $\Mt_\mC$ has $m \times 4^n + 6$ truth-values and $\Sigma_\mC$ has $3+m$ connectives. Therefore, $(\Mt_\mC)_{\mathsf{m}}$ has
    $m \times 4^n + 7$ truth-values and $(\Sigma_{\mathcal{C}})_{\mathsf{m}}$ has $m \times 4^n + m + 9$ connectives. We can therefore conclude the proof just by 
    invoking Theorem~\ref{thm:Halt}.
\end{proof}

As a simple corollary we obtain that following result about Nmatrices, or better, about their underlying multi-algebras.

\begin{corollary} 
The problem of generating all expressible unary multi-functions in an arbitrary finite Nmatrix is not computable.
\end{corollary}
\begin{proof}
    Just note that if we could compute the set of all expressible unary multi-functions, as the set is necessarily finite, we could test each of them
    for the separation of values, as illustrated in the case of matrices in Example~\ref{ex:luk2}.
\end{proof}

\section{Conclusion}\label{sec:conclude}

In this paper we have shown that, contrarily to the most common case of logical matrices, the monadicity property is undecidable for non-deterministic 
matrices. As a consequence, we conclude that the set of all multi-functions expressible in a given finite Nmatrix is not computable, in general. These 
results, of course, do not spoil the usefulness of the techniques for obtaining axiomatizations, analytical calculi and automated proof-search for 
monadic non-deterministic matrices. 
This is especially the case since, for a given Nmatrix, one can always define a monadic Nmatrix over an enriched signature, such that its logic is a conservative
extension of the logic of the previous Nmatrix, as described in~\cite{synthese}.
The results show, however, that tool support for logics based on non-deterministic matrices must necessarily have its limitations.

On a closer perspective, the reduction we have obtained from counter machines to Nmatrices (of which non-determinism is a fundamental ingredient) just 
adds to the initial perception that allowing for non-determinism brings a substantial amount of expressive power to logical matrices. Concretely, it opens the door 
for studying the computational hardness of other fundamental meta-theoretical questions regarding logics defined by Nmatrices. In particular, we will be 
interested in studying the problem of deciding whether two given finite Nmatrices define the same logic, a fundamental question raised by Zohar and Avron 
in~\cite{AvronZohar}, for which only necessary or sufficient conditions are known.

Additionally, we deem it important to further explore the connections between Nmatrices and term-dag-automata (an interesting computational model for 
term languages~\cite{AutomataOnDAGRepresentationsOfFiniteTrees,ClosureProperties})
and which informed our undecidability result. Another relevant direction for further investigation is the systematic study of infectious semantics, in 
the lines of~\cite{ismvl,infect}, whose variable inclusion properties also played an important role in our results.

\bibliographystyle{eptcs}
\bibliography{biblio}

\begin{thebibliography}{10}
\providecommand{\bibitemdeclare}[2]{}
\providecommand{\surnamestart}{}
\providecommand{\surnameend}{}
\providecommand{\urlprefix}{Available at }
\providecommand{\url}[1]{\texttt{#1}}
\providecommand{\href}[2]{\texttt{#2}}
\providecommand{\urlalt}[2]{\href{#1}{#2}}
\providecommand{\doi}[1]{doi:\urlalt{http://dx.doi.org/#1}{#1}}
\providecommand{\eprint}[1]{arXiv:\urlalt{https://arxiv.org/abs/#1}{#1}}
\providecommand{\bibinfo}[2]{#2}

\bibitemdeclare{article}{ClosureProperties}
\bibitem{ClosureProperties}
\bibinfo{author}{Siva \surnamestart Anantharaman\surnameend},
  \bibinfo{author}{Paliath \surnamestart Narendran\surnameend} \&
  \bibinfo{author}{Michael \surnamestart Rusinowitch\surnameend}
  (\bibinfo{year}{2005}): \emph{\bibinfo{title}{Closure properties and decision
  problems of dag automata}}.
\newblock {\sl \bibinfo{journal}{Information Processing Letters}}
  \bibinfo{volume}{94}(\bibinfo{number}{5}), pp. \bibinfo{pages}{231--240},
  \doi{10.1016/j.ipl.2005.02.004}.

\bibitemdeclare{article}{AKZ}
\bibitem{AKZ}
\bibinfo{author}{Arnon \surnamestart Avron\surnameend}, \bibinfo{author}{Beata
  \surnamestart Konikowska\surnameend} \& \bibinfo{author}{Anna \surnamestart
  Zamansky\surnameend} (\bibinfo{year}{2013}): \emph{\bibinfo{title}{{C}ut-free
  {S}equent {C}alculi for {C}-systems with {G}eneralized {F}inite-valued
  {S}emantics}}.
\newblock {\sl \bibinfo{journal}{{J}ournal of {L}ogic and {C}omputation}}
  \bibinfo{volume}{23}(\bibinfo{number}{3}), pp. \bibinfo{pages}{517--540},
  \doi{10.1093/logcom/exs039}.

\bibitemdeclare{article}{Avron0}
\bibitem{Avron0}
\bibinfo{author}{Arnon \surnamestart Avron\surnameend} \& \bibinfo{author}{Iddo
  \surnamestart Lev\surnameend} (\bibinfo{year}{2005}):
  \emph{\bibinfo{title}{Non-deterministic multiple-valued structures}}.
\newblock {\sl \bibinfo{journal}{Journal of Logic and Computation}}
  \bibinfo{volume}{15}(\bibinfo{number}{3}), pp. \bibinfo{pages}{241--261},
  \doi{10.1093/logcom/exi001}.

\bibitemdeclare{incollection}{Avron1}
\bibitem{Avron1}
\bibinfo{author}{Arnon \surnamestart Avron\surnameend} \& \bibinfo{author}{Anna
  \surnamestart Zamansky\surnameend} (\bibinfo{year}{2011}):
  \emph{\bibinfo{title}{Non-deterministic semantics for logical systems}}.
\newblock In \bibinfo{editor}{Dov~M. \surnamestart Gabbay\surnameend} \&
  \bibinfo{editor}{Franz \surnamestart Guenthner\surnameend}, editors: {\sl
  \bibinfo{booktitle}{Handbook of Philosophical Logic}}, \bibinfo{volume}{16},
  \bibinfo{publisher}{Springer}, pp. \bibinfo{pages}{227--304},
  \doi{10.1007/978-94-007-0479-4\_4}.

\bibitemdeclare{article}{AvronZohar}
\bibitem{AvronZohar}
\bibinfo{author}{Arnon \surnamestart Avron\surnameend} \& \bibinfo{author}{Yoni
  \surnamestart Zohar\surnameend} (\bibinfo{year}{2019}):
  \emph{\bibinfo{title}{Rexpansions of non-deterministic matrices and their
  applications}}.
\newblock {\sl \bibinfo{journal}{The Review of Symbolic Logic}}
  \bibinfo{volume}{12}(\bibinfo{number}{1}), pp. \bibinfo{pages}{173--200},
  \doi{10.1017/S1755020318000321}.

\bibitemdeclare{article}{infect}
\bibitem{infect}
\bibinfo{author}{Stefano \surnamestart Bonzio\surnameend},
  \bibinfo{author}{Tommaso \surnamestart Moraschini\surnameend} \&
  \bibinfo{author}{Michele~Pra \surnamestart Baldi\surnameend}
  (\bibinfo{year}{2020}): \emph{\bibinfo{title}{Logics of left variable
  inclusion and {P}{\l}onka sums of matrices}}.
\newblock {\sl \bibinfo{journal}{{A}rchive for {M}athematical {L}ogic}}
  \bibinfo{volume}{60}, pp. \bibinfo{pages}{49--76},
  \doi{10.1007/s00153-020-00727-6}.

\bibitemdeclare{incollection}{wollic19}
\bibitem{wollic19}
\bibinfo{author}{Carlos \surnamestart Caleiro\surnameend} \&
  \bibinfo{author}{S{\'e}rgio \surnamestart Marcelino\surnameend}
  (\bibinfo{year}{2019}): \emph{\bibinfo{title}{Analytic calculi for monadic
  {PN}matrices}}.
\newblock In \bibinfo{editor}{Rosalie \surnamestart Iemhoff\surnameend},
  \bibinfo{editor}{Michael \surnamestart Moortgat\surnameend} \&
  \bibinfo{editor}{Ruy \surnamestart de~Queiroz\surnameend}, editors: {\sl
  \bibinfo{booktitle}{{L}ogic, {L}anguage, {I}nformation, and {C}omputation
  ({W}o{L}{L}{I}{C} 2019)}}, {\sl \bibinfo{series}{LNCS}}
  \bibinfo{volume}{11541}, \bibinfo{publisher}{Springer-Verlag}, pp.
  \bibinfo{pages}{84--98}, \doi{10.1007/978-3-662-59533-6\_6}.

\bibitemdeclare{inproceedings}{ismvl}
\bibitem{ismvl}
\bibinfo{author}{Carlos \surnamestart Caleiro\surnameend},
  \bibinfo{author}{S{\'e}rgio \surnamestart Marcelino\surnameend} \&
  \bibinfo{author}{Pedro \surnamestart Filipe\surnameend}
  (\bibinfo{year}{2020}): \emph{\bibinfo{title}{Infectious semantics and
  analytic calculi for even more inclusion logics}}.
\newblock In: {\sl \bibinfo{booktitle}{2020 {IEEE} 50th {I}nternational
  {S}ymposium on {M}ultiple-{V}alued {L}ogic}}, pp. \bibinfo{pages}{224--229},
  \doi{10.1109/ISMVL49045.2020.000-1}.

\bibitemdeclare{article}{CCALJM}
\bibitem{CCALJM}
\bibinfo{author}{Carlos \surnamestart Caleiro\surnameend},
  \bibinfo{author}{Jo{\~{a}}o \surnamestart Marcos\surnameend} \&
  \bibinfo{author}{Marco \surnamestart Volpe\surnameend}
  (\bibinfo{year}{2015}): \emph{\bibinfo{title}{Bivalent semantics, generalized
  compositionality and analytic classic-like tableaux for finite-valued
  logics}}.
\newblock {\sl \bibinfo{journal}{{T}heoretical {C}omputer {S}cience}}
  \bibinfo{volume}{603}, pp. \bibinfo{pages}{84--110},
  \doi{10.1016/j.tcs.2015.07.016}.

\bibitemdeclare{techreport}{AutomataOnDAGRepresentationsOfFiniteTrees}
\bibitem{AutomataOnDAGRepresentationsOfFiniteTrees}
\bibinfo{author}{Witold \surnamestart Charatonik\surnameend}
  (\bibinfo{year}{1999}): \emph{\bibinfo{title}{{A}utomata on {DAG}
  {R}epresentations of {F}inite {T}rees}}.
\newblock \bibinfo{type}{Technical Report} \bibinfo{number}{MPI-I-1999-2-001},
  \bibinfo{institution}{{M}ax-{P}lanck-{I}nstitut f{\"u}r {I}nformatik},
  \bibinfo{address}{Saarbr{\"u}cken, Germany}.

\bibitemdeclare{book}{AlgLogBook}
\bibitem{AlgLogBook}
\bibinfo{author}{Josep \surnamestart Font\surnameend} (\bibinfo{year}{2016}):
  \emph{\bibinfo{title}{{A}bstract {A}lgebraic {L}ogic}}.
\newblock {\sl \bibinfo{series}{{M}athematical {L}ogic and
  {F}oundations}}~\bibinfo{volume}{60}, \bibinfo{publisher}{{C}ollege
  {P}ublications}.

\bibitemdeclare{book}{Lau}
\bibitem{Lau}
\bibinfo{author}{Dietlinde \surnamestart Lau\surnameend}
  (\bibinfo{year}{2006}): \emph{\bibinfo{title}{Function Algebras on Finite
  Sets, A Basic Course on Many-Valued Logic and Clone Theory}}.
\newblock \bibinfo{series}{Springer Monographs in Mathematics},
  \bibinfo{publisher}{Springer}, \doi{10.1007/3-540-36023-9}.

\bibitemdeclare{incollection}{wollic17}
\bibitem{wollic17}
\bibinfo{author}{S{\'e}rgio \surnamestart Marcelino\surnameend} \&
  \bibinfo{author}{Carlos \surnamestart Caleiro\surnameend}
  (\bibinfo{year}{2017}): \emph{\bibinfo{title}{Disjoint fibring of
  non-deterministic matrices}}.
\newblock In \bibinfo{editor}{Juliette \surnamestart Kennedy\surnameend} \&
  \bibinfo{editor}{Ruy~J.G.B. \surnamestart de~Queiroz\surnameend}, editors:
  {\sl \bibinfo{booktitle}{{L}ogic, {L}anguage, {I}nformation, and
  {C}omputation ({W}o{L}{L}{I}{C} 2017)}}, {\sl \bibinfo{series}{LNCS}}
  \bibinfo{volume}{10388}, \bibinfo{publisher}{Springer-Verlag}, pp.
  \bibinfo{pages}{242--255}, \doi{10.1007/978-3-662-55386-2\_17}.

\bibitemdeclare{article}{synthese}
\bibitem{synthese}
\bibinfo{author}{S{\'e}rgio \surnamestart Marcelino\surnameend} \&
  \bibinfo{author}{Carlos \surnamestart Caleiro\surnameend}
  (\bibinfo{year}{2019}): \emph{\bibinfo{title}{Axiomatizing non-deterministic
  many-valued generalized consequence relations}}.
\newblock {\sl \bibinfo{journal}{Synthese}} \bibinfo{volume}{198}, pp.
  \bibinfo{pages}{5373--5390}, \doi{10.1007/s11229-019-02142-8}.

\bibitemdeclare{book}{minsky}
\bibitem{minsky}
\bibinfo{author}{Marvin \surnamestart Minsky\surnameend}
  (\bibinfo{year}{1967}): \emph{\bibinfo{title}{{C}omputation: {F}inite and
  {I}nfinite {M}achines}}.
\newblock \bibinfo{publisher}{{P}rentice-{H}all}.

\bibitemdeclare{book}{SS}
\bibitem{SS}
\bibinfo{author}{David~John \surnamestart Shoesmith\surnameend} \&
  \bibinfo{author}{Timothy \surnamestart Smiley\surnameend}
  (\bibinfo{year}{1978}): \emph{\bibinfo{title}{Multiple-conclusion logic}}.
\newblock \bibinfo{publisher}{Cambridge University Press},
  \doi{10.1017/CBO9780511565687}.

\bibitemdeclare{book}{WojBook}
\bibitem{WojBook}
\bibinfo{author}{Ryszard \surnamestart W{\'o}jcicki\surnameend}
  (\bibinfo{year}{1988}): \emph{\bibinfo{title}{Theory of Logical Calculi}}.
\newblock {\sl \bibinfo{series}{Synthese Library}} \bibinfo{volume}{199},
  \bibinfo{publisher}{Kluwer}, \doi{10.1007/978-94-015-6942-2}.

\end{thebibliography}
\end{document}